\documentclass[lettersize,journal]{IEEEtran}
\usepackage{amsmath,amsfonts}
\usepackage{algorithmic}
\usepackage{algorithm}
\usepackage{array}
\usepackage[caption=false,font=normalsize,labelfont=sf,textfont=sf]{subfig}
\usepackage{textcomp}
\usepackage{stfloats}
\usepackage{url}
\usepackage{verbatim}
\usepackage{graphicx}
\usepackage{cite}
\usepackage{amsthm}
\newtheorem{theorem}{Theorem}[section]
\newtheorem{lemma}[theorem]{Lemma}
\usepackage{mathabx}
\usepackage{slashbox}
\usepackage{xspace}
\newcommand{\ie}{{\it i.e., }}
\newcommand{\eg}{{\it e.g., }}
\newcommand{\etal}{{\it et~al. }}
\newcommand{\SEE}{HEET\xspace}
\newcommand{\SEEL}{Homogeneous Equivalent Execution Time}

\begin{document}

\title{HEET: A Heterogeneity Measure to Quantify the Difference across Distributed Computing Systems}

\author{
Ali Mokhtari \IEEEmembership{Student Member, IEEE}, 
Saeid Ghafouri \IEEEmembership{Student Member, IEEE}, 
Pooyan Jamshidi \IEEEmembership{Member, IEEE}, and 
Mohsen Amini Salehi \IEEEmembership{Member, IEEE}
\thanks{This work was supported by the National Science Foundation under awards\# CNS-2007209, 2007202, 2107463, 2233873, and 2038080. (\textit{Corresponding authors: Ali Mokhtari, Mohsen Amini Salehi})}
\thanks{Ali Mokhtari is with the School of Computing and Informatics, University of Louisiana at Lafayette, Louisiana, USA. Email: ali.mokhtari1@louisiana.edu}
\thanks{Saeid Ghafouri is with Queen Mary University of London, London, UK. Email: s.ghafouri@qmul.ac.uk}
\thanks{Pooyan Jamshidi is with the Computer Science and Engineering Department, University of South Carolina, South Carolina, USA. Email: pjamshid@cse.sc.edu}
\thanks{Mohsen Amini Salehi was with the School of Computing and Informatics, University of Louisiana at Lafayette, Louisiana, USA. He is now with the Computer Science and Engineering Department, University of North Texas, Texas, USA. Email: mohsen.aminisalehi@unt.edu}
}

\markboth{}
{Ali Mokhtari, Saeid Ghafouri \MakeLowercase{\textit{et al.}}: HEET: A Heterogeneity Measure to Quantify the Difference across Distributed Computing Systems}


\maketitle

\begin{abstract}


Although system heterogeneity has been extensively studied in the past, there is yet to be a study on measuring the impact of heterogeneity on system performance. For this purpose, we propose a heterogeneity measure that can characterize the impact of the heterogeneity of a system on its performance behavior in terms of throughput or makespan. We develop a mathematical model to characterize a heterogeneous system in terms of its task and machine heterogeneity dimensions and then reduce it to a single value, called Homogeneous Equivalent Execution Time (HEET), which represents the execution time behavior of the entire system. We used AWS EC2 instances to implement a real-world machine learning inference system. Performance evaluation of the HEET score across different heterogeneous system configurations demonstrates that HEET can accurately characterize the performance behavior of these systems. In particular, the results show that our proposed method is capable of predicting the true makespan of heterogeneous systems without online evaluations with an average precision of 84\%. This heterogeneity measure is instrumental for solution architects to configure their systems proactively to be sufficiently heterogeneous to meet their desired performance objectives. 


\end{abstract}


\section{Introduction} 
\subsection{Research Motivations and Goals}
Heterogeneity has been an indispensable aspect of distributed computing throughout the history of these systems. 
In the modern era, as Moore's law is losing momentum due to power density and heat dissipation limitations~\cite{taylor2012dark, esmaeilzadeh2011dark}, heterogeneous computing systems have attracted even more attention to overcome the slowdown in Moore's law and fulfill the desire for higher performance in various types of distributed systems. In particular, with the ubiquity of accelerators (\eg GPUs and TPUs) and domain-specific computing (through ASICs \cite{taylor2020asic} and FPGA \cite{bobda2022future}), the matter of heterogeneity and harnessing it has become a more critical challenge than ever before to deal with. 

The footprint of heterogeneity can be traced in all forms of distributed systems. Hyperscaler cloud providers, such as AWS and Microsoft Azure, offer and operate based on a wide range of ``machine types'', ranging from general-purpose X86 and ARM machines to FPGAs and accelerators. Cloud users can use this heterogeneity to mitigate their costs and improve QoS. For example, Amazon SageMaker~\cite{Amazon_SageMaker} operates based on heterogeneous cloud machines to build, train, and deploy machine learning (ML) models. As an example, ResNet-50 training in a heterogeneous cluster with GPUs (\texttt{ml.g5.xlarge}) and compute-optimized (\texttt{ml.c5n.2xlarge}) machines yields 13\% lower cost than in a homogeneous cluster with only \texttt{ml.g5.xlarge} GPUs~\cite{ASageMaker_Resnet}. In edge computing, domain-specific accelerators (ASICs and FPGAs) and general-purpose processors are commonly used to perform near-data real-time processing. For example, Google smartglasses Enterprise Edition 2 \cite{google_glass} is equipped with a System-on-Chip (SoC) that includes a multicore ARM CPU, a GPU, and a Qualcomm AI Engine to provide onboard computer vision~\cite{google_glass, Qualcomm}. 
In the HPC context, the deployment of architecturally heterogeneous machines has become prevalent to fulfill the power and performance desires~\cite{cardwell2020truly}. As just one example, the HPE Cray EX architecture combines third-generation AMD EPYC CPUs (optimized for HPC and AI) with AMD Instinct 250X accelerators~\cite{TOP500}. 

As noted above, heterogeneity is key to improving various performance objectives of distributed systems, such as cost, energy consumption, and throughput. Therefore, harnessing system heterogeneity has been a long-standing goal in distributed systems. Previous research works have predominantly aimed to optimize a certain performance metric (\eg energy consumption, deadline meet rate, or throughput) with respect to the heterogeneity of the underlying computing systems. Examples of such work are task scheduling \cite{mokhtari2022felare,denninnart2020efficient}, load balancing \cite{perez2021sigmoid,khalid2019troodon,nozal2020enginecl}, and cloud elasticity \cite{vandebon2019enhanced, zhong2020cost}. However, to our knowledge, \emph{ there is no concrete study of the performance of these works when changing the degree (level) of heterogeneity in the underlying system}. That is, the impact of the same system-level solutions (\eg cloud elasticity, scheduling methods) between computing systems with various degrees of heterogeneity is unknown.

\begin{figure*}[!t]
\centering
\subfloat[]{\includegraphics[width=0.43\textwidth]{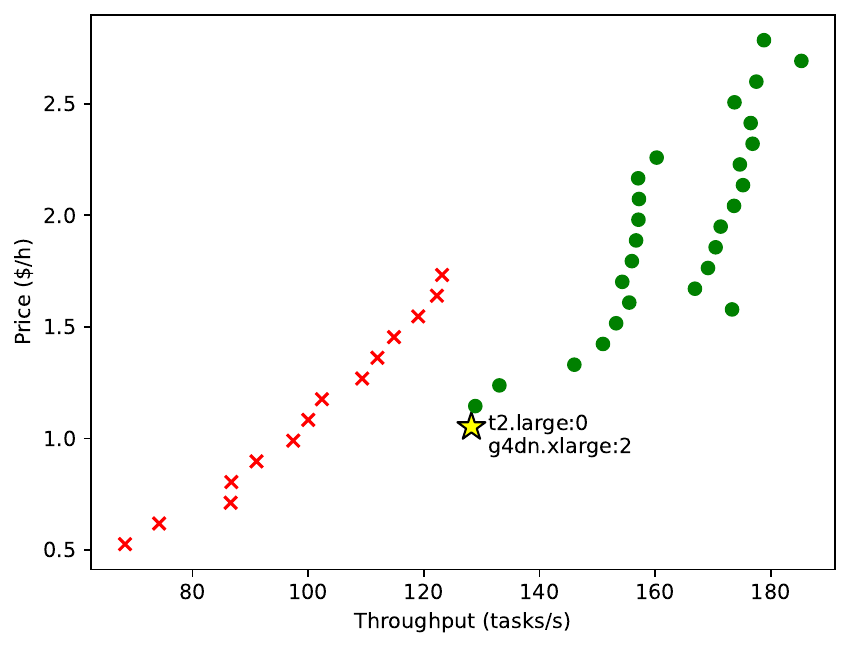}%
\label{fig:cost_opt_a}}
\hfil
\subfloat[]{\includegraphics[width=0.43\textwidth]{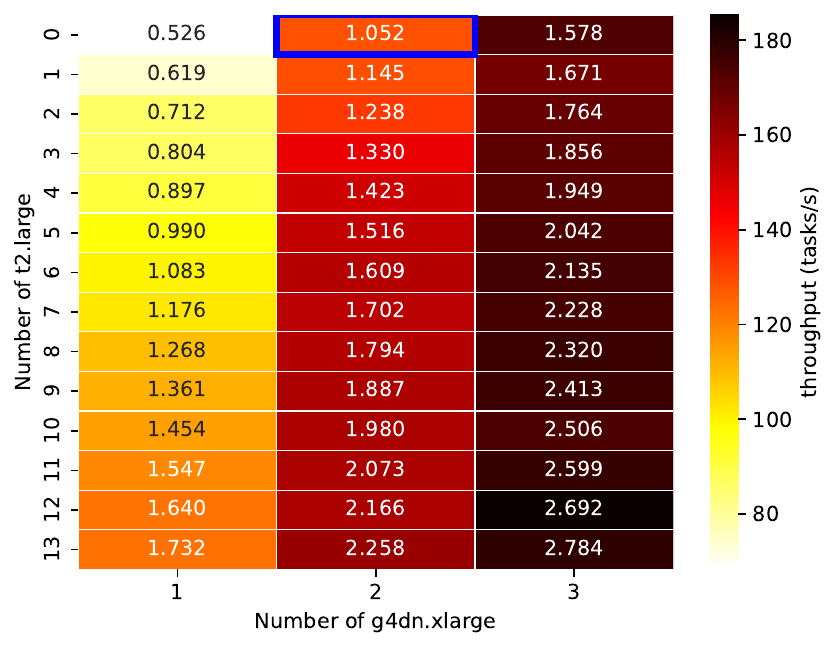}%
\label{fig:cost_opt_b}}
\caption{Illustrating the optimal system configuration that minimizes the cost while meeting the throughput target for a workload of 1000 tasks of speech recognition, image classification, object detection, and question-answering types. Different numbers of AWS EC2 instances of \texttt{t2.large} and \texttt{g4dn.xlarge} types are used to form heterogeneous system configurations. In Figure~\ref{fig:cost_opt_a}, the green points represent the system configurations that meet the throughput target ($\ge$125 tasks/sec), and the red points violate this constraint. The golden star point is the optimal configuration that meets the throughput target with minimum cost. In Figure~\ref{fig:cost_opt_b}, numbers in each cell show the cost of each configuration, and the colors illustrate the performance. Darker red represents higher throughput. The cell with a blue frame is the optimal configuration.}
\label{fig:cost_opt}
\vspace{-4mm}
\end{figure*}

In the context of cloud elasticity, for a given cloud-based application, solution architects need to know the implications of modifying (\ie adding, removing, or replacing) the allocated machines on the application performance in terms of throughput and cost incurred. The system architects must answer questions such as ``For a certain arriving workload pattern, what type of machine(s) must be allocated to the existing machines so that the user-defined performance objectives (\eg cost and throughput) are satisfied?'' They need the ability to \emph{compare} performances of different (potentially) heterogeneous systems in terms of their performance objectives (\eg incurred costs and throughput constraint) before committing their allocation decisions. Figure \ref{fig:cost_opt} shows the throughput of different heterogeneous systems, featuring varying numbers of AWS EC2 instances of \texttt{t2.large} and \texttt{g4dn.xlarge} types, along with the corresponding incurred cost. In Figure \ref{fig:cost_opt_a}, the configurations (\ie a certain number of \texttt{t2.large} and \texttt{g4dn.xlarge} instances) that meet the throughput constraint (\ie 125 tasks/sec) are shown in green, and those that do not satisfy the throughput constraint are in red. While there are multiple viable configurations that meet the throughput constraint, there is an optimal configuration that minimizes the incurred cost. Figure \ref{fig:cost_opt_b} also shows the number of instances with the corresponding price. These observations indicate that to minimize the cost incurred by different heterogeneous computing systems, we need to be able to quantify the impact of heterogeneity on the performance metric (\eg throughput), then we can configure a heterogeneous system with the minimum cost. Consequently, in this study, we introduce a novel heterogeneity measure aimed at making the performance (\ie throughput) of heterogeneous systems predictable. Performance predictability empowers solution architects to estimate the throughput of various system configurations offline. Consequently, they can identify the optimal configuration that minimizes cost while meeting the throughput target.

The influence of system heterogeneity on system-level solutions calls for a measurable metric to quantify the heterogeneity of a computing system. Although heterogeneous computing has been extensively investigated in the past, there is yet to be a concrete measure to quantify the heterogeneity of a system and characterize its performance so that different heterogeneous systems become comparable. As such, \emph{our goal in this study is to propose a `` performance-driven heterogeneity measure'' that can characterize the impact of the heterogeneity level of a system on its performance objective (\ie throughput) and make the system comparable with its counterparts}. 

\subsection{Problem Statement and Summary of Contributions}
We define a heterogeneous computing system as a set of architecturally diverse machines that work together to complete a set of requests (a.k.a. \emph{tasks}) with different computational requirements. We categorize the tasks that arrive at a system based on the type of operation they perform and call them \emph{task types}. For example, in a system that helps blind and visually impaired people \cite{mokhtari2022felare,zobaed2022edge}, the task types can be obstacle detection, face recognition, and speech recognition. Moreover, we classify machines of a computing system based on their architectural and performance characteristics and call each one a \emph{machine type}. For example, a cloud solution architect can form a virtual compute cluster using ARM, x86-based, and GPU machine types. In this work, we consider the system heterogeneity that emanates from the diversity in machine types and computational requirements of task types. That is, system heterogeneity has two dimensions: (i) \emph{machine heterogeneity} and (ii) \emph{task heterogeneity}. 
Profiling various task types on heterogeneous machine types can describe their execution time behavior.
Variations in the performance (a.k.a. execution time) of a given task type across all the machine types are defined as machine heterogeneity, whereas variations in the execution time of different task types on a given machine type are defined as task heterogeneity. Then, the \emph{system heterogeneity} is defined as the compounded heterogeneity of these two dimensions.

To quantify the heterogeneity of such computing systems in a manner that can describe the performance behavior of the system in terms of makespan and throughput, we need to address the following research questions:
(i) \emph{How to find a measure that can quantify the heterogeneity of any given computing system?} (ii) \emph{How to exploit the measure of heterogeneity to predict the performance behavior of a computing system?}

We define \emph{Expected Execution Time (EET)} as a matrix to store the expected execution time of each type of task on each type of machine. In this way, the entry $EET[i,j]$ (denoted $e_{ij}$) represents the expected execution time of the task type $i$ in machine type $j$. For systems with multiple instances of the same machine type, the corresponding columns of those machines in the EET matrix are identical. The EET matrix, as a whole, represents the expected performance of the entire system in terms of the execution times of the tasks. As such, the matrix can be used as a guide to understand the throughput that the system can potentially achieve. We propose a mathematical model to measure the task heterogeneity and machine heterogeneity levels separately as two dimensions describing the system heterogeneity level. Then, we devise a mathematical model to determine the overall system heterogeneity based on these dimensions. For this purpose, we examine various measures of central tendency, namely arithmetic, geometric, and harmonic means, and propose a single measure, called \SEEL~ (\emph{\SEE}), which, for a given EET matrix and workload, describes the expected execution time of a hypothetical homogeneous system whose throughput is similar to the heterogeneous system represented by the EET matrix. Subsequently, we employ the \SEE measure to determine the throughput of the heterogeneous computing system. We evaluated the measure \SEE for various heterogeneous systems and workloads and showed that it can accurately describe the impact of the level of heterogeneity on the desired throughput. 

In summary, the specific contributions of this paper are as follows:
\begin{itemize}
     \item Providing a measure to quantify the system heterogeneity such that it can be used to determine the throughput of the system for a given workload.
     
     \item Proposing a systematic approach to analyzing heterogeneity of a computing system by means of decoupling the heterogeneity into machine and task heterogeneity dimensions and characterizing each dimension separately by making use of a custom-defined speedup metric.
        
    \item Proving the appropriateness of the arithmetic mean and harmonic mean to measure the central tendency of speedup values due to machine heterogeneity and task heterogeneity, respectively.     
    
    \item Validating how system performance (makespan and throughput) can be derived as a function of the proposed heterogeneity measure: \url{https://github.com/hpcclab/heterogeneity_measure}. 
\end{itemize}

The rest of this paper is organized as follows. In Section~\ref{method}, we describe the mathematical model to quantify the heterogeneity of the system based on the overall speedup that a system obtains by employing heterogeneity. Section \ref{system-design} provides an overview of the implementation in the real world that we have developed to validate the accuracy of the \SEE measure. Section~\ref{experiments}, examines the proposed heterogeneity measure. In Section~\ref{related_works}, we review the previous studies and position our work with respect to them. Conclusions and future works are discussed in Section~\ref{conclusion}.

\section{Analyzing Performance Characteristics of Heterogeneous Systems}\label{method}
\subsection{overview}

System heterogeneity is the result of the synergy between machine heterogeneity and task heterogeneity dimensions. Accordingly, our approach to quantifying system heterogeneity is to measure the heterogeneity of each dimension individually and then combine them to quantify the overall system heterogeneity. To this end, we base our analysis on the notion of EET matrix that is representative of the system performance. Considering that each row of EET represents a task type and each column represents a machine type, the row-wise variations illustrate the machine heterogeneity, and similarly, the column-wise variations express the task heterogeneity. Note that, in the event that some machines in the system are homogeneous, their corresponding columns in the EET matrix are repeated.

We exploit the notion of speed-up to characterize the impact of heterogeneity on the system execution time behavior. The gained ``speed-up due to machine heterogeneity'' can be described by the row-wise analysis of the EET matrix. We represent this speed-up by a row vector, denoted $\Vec{\alpha}^{(i)}$, which has the same dimension as the i\textsuperscript{th} row of the EET matrix. Each entry $\alpha_j^{(i)}$ denotes the speed-up that the system can achieve when the machine type $j$ executes the task type $i$ instead of running it on the slowest machine type. Consequently, given the EET matrix of size $m\times n$, the value of $\alpha_j^{(i)}$ is determined based on Equation~\ref{eq:sm}. 
\begin{equation} \label{eq:sm}
    \alpha_j^{(i)} = \frac{\max\limits_{j=1}^n e_{ij}}{e_{ij}}
\end{equation}

In the above equation, $\alpha_j^{(i)} = 1$, if and only if machine type $j$ is the same as the slowest machine for task type $i$. 

Similarly, we define ``speed-up due to task heterogeneity'', denoted $\Vec{\beta}_{(j)}$, as a column vector with the same dimension of j\textsuperscript{th} column of the EET matrix. The entry $\beta^i_{(j)}$ indicates the speed-up achieved by executing a task of type $i$ on the machine type $j$, as opposed to executing the slowest task type on that machine. Formally, $\beta^i_{(j)}$ is calculated on the basis of Equation~\ref{eq:st}. 

\begin{equation} \label{eq:st}
    \beta^i_{(j)} = \frac{\max\limits_{i=1}^m e_{ij}}{e_{ij}}
\end{equation}

The value of $\beta^i_{(j)}$ is one if and only if the task type $i$ is the slowest task type on the machine type $j$. 

Given the speed-up vectors due to machine heterogeneity and task heterogeneity, we need to represent each speed-up vector in the form of a scalar value. This scalar value for $\Vec{\alpha}^i$ ($\Vec{\beta}_{(j)}$) accurately describes how much all machine types (task types) together can gain speed-up due to heterogeneity in machines (tasks). Based on the mean-field method~\cite{allmeier2022mean, ying2016approximation}, the interaction of variables in a complex stochastic system can be replaced by the average interactions between these variables. As such, we can use the mean to represent the speed-up behavior of all $(task, machine)$ pairs in both $\Vec{\alpha}^{(i)}$ and $\Vec{\beta}_{(j)}$. For this purpose, we employ statistical measures of central tendency (arithmetic, geometric, and harmonic mean) to accurately represent $\Vec{\alpha}^{(i)}$ and $\Vec{\beta}_{(j)}$. However, the challenge is that there is no consensus on the appropriateness of these measures to capture the central tendency of a specific use case \cite{lilja2005measuring, john2004more} and it must be investigated case by case. An appropriate mean speed-up value derived from the EET matrix must precisely depict the ``real speed-up'', a.k.a. \emph{true speed-up} (denoted $\Gamma$), that the heterogeneous system can achieve for a given workload. We exploit the notion of makespan (\ie the total time of executing the workload) to calculate the true speed-up. According to Equation~\ref{eq:gamma}, the true speed-up is determined based on the makespan of executing the workload on the heterogeneous system with respect to its slowest ``counterpart homogeneous system'', as the base system (a.k.a. baseline). 

\begin{equation} \label{eq:gamma}
    \Gamma = \frac{\text{homogeneous system makespan}}{\text{ heterogeneous system makespan}}
\end{equation}

Note that the counterpart homogeneous system is represented by an EET matrix whose entries are all equal to the maximum value of the EET matrix of the heterogeneous system. We use $\Gamma_M$ and $\Gamma_T$ to represent the true speed-up with respect to machine heterogeneity and task heterogeneity, respectively. 

Once we know the true speed-up for a given workload, we can compare it against the calculated mean speed-up to know how accurate the calculated one is. We provide three lemmas to introduce appropriate central tendency measures that can represent speed-up due to machine heterogeneity and task heterogeneity. 

In the first step, to quantify the heterogeneity of the system, for each type of machine, we utilize the EET matrix to generate speed-up vectors due to the heterogeneity of the task using Equation~\ref{eq:st}. These speed-up column vectors construct a speed-up matrix due to task heterogeneity, denoted by$S^T$. In the second step, we represent each speed-up vector due to task heterogeneity by its mean value. Note that, in this step, we fuse all task types into a single hypothetical equivalent task type that can describe the execution time behavior of all task types. Taking into account that we consider the execution time of the slowest task type as the baseline, in the third step, for each machine type, we utilize the mean speed-up values due to task heterogeneity and the baseline execution time to determine the execution time behavior of the hypothetical equivalent task type. We repeat this step for all machine types to construct a row vector depicting the execution times of the hypothetical equivalent task type on the machine types. In step four, based on the row vector generated in the last step (as the EET of the representative task type), we generate the speed-up vector due to machine heterogeneity based on Equation~\ref{eq:sm}. Eventually, we reduce the speed-up vector due to machine heterogeneity for the hypothetical equivalent task type into a single scalar value. We use this mean speed-up value and the execution time of the slowest machine for the hypothetical task type to determine a single-value execution time that represents the execution time behavior of the entire system. We use this representative execution time of the hypothetical equivalent task type on the hypothetical equivalent machine type to define a heterogeneity measure.

Provided that heterogeneity affects system performance by means of execution time (reflected in the EET matrix), we define a measure based on execution time behavior, called \emph{\SEEL \xspace (\SEE)}, to estimate the impact of the system heterogeneity on throughput. In fact, \SEE represents how fast the heterogeneous machines in the system process heterogeneous tasks.
Our hypothesis is that systems with the same \SEE measure expose comparable execution time behavior, thus, for a given workload, they exhibit similar makespan and throughput. Therefore, the \SEE score is capable of properly characterizing the heterogeneity of the system with respect to the throughput it can offer without running the workload against the heterogeneous system. In the rest of this section, we elaborate on characterizing machine heterogeneity and task heterogeneity dimensions. Then, we discuss how to fuse these dimensions to characterize the system heterogeneity.

\subsection{Characterizing Machine Heterogeneity} \label{sec:charac_machine}
For the task type $i$, we process row $i$ of the EET matrix with respect to the slowest machine for that task type to form a \emph{row speed-up vector}, denoted $\Vec{\alpha}^{((i)}$, representing the speed-up values due to machine heterogeneity. For the task type $i$, we use the central tendency measure (mean) of the vector components of the row speed-up, denoted $\widebar{\alpha}^{(i)}$, to aggregate the speed-up behavior of all types of machines. 

According to \cite{smith1988characterizing}, in the circumstances where performance is expressed as a rate (\eg flops), generally the harmonic mean can accurately express the central tendency. Also, the central tendency can usually be represented by the arithmetic mean when the performance is of a time nature (\eg makespan or total execution time of a benchmark). Lastly, they suggest avoiding using geometric mean when the performance is of the time or rate nature. In another study~\cite{john2004more}, the speed-up is considered as the performance metric to compare an improved system against a baseline one. To this end, they used a benchmark suite to evaluate the performance of each system. Then, based on the makespan of each individual benchmark in the benchmark suite, the speed-up for the enhanced system is calculated. Next, the authors discussed different measures of central tendency (\ie arithmetic, harmonic, and geometric mean) to summarize the speed-up results of the benchmarks into a single number such that it appropriately describes the overall speed-up for the entire benchmark suite. To validate the suitability of the mean speed-up measure, they compared it with the speed-up achieved by considering the makespan of the entire benchmark suite on both systems.

In this research, we also follow the same approach as \cite{john2004more} to validate the accuracy of the central tendency measure of the speed-up matrix. In particular, we use the notion of true speed-up ($\Gamma$) to verify that the central tendency measure accurately represents the speed-up vector of the row. In addition, the intensity of task arrival to the system impacts machines' utilization, which, in turn, affects the true speed-up.
In Lemma~\ref{lemma_sm_avg}, we study an extreme case where the task arrival rate is large enough such that all machines in the system have a task for execution at all times. We show that for such a system, the arithmetic mean can appropriately summarize $\Vec{\alpha}^{((i)}$ and represent the mean speed-up due to machine heterogeneity. For the other side of the spectrum, where the task arrival rate is low such that only one machine in the system is executing a task at a time, we present Lemma~\ref{lemma_sm_h} to prove that the harmonic mean should be used to accurately summarize $\Vec{\alpha}^{((i)}$ and represent the mean speed-up due to machine heterogeneity. In these lemmas, we assume that there is a single unbounded FCFS queue of tasks that are all available for execution (like the bag-of-tasks \cite{oxley2014makespan}). Whenever a machine becomes free, it takes the next task from the queue to execute it.

\begin{lemma}\label{lemma_sm_avg}
    Let $EET=[e_{ij}]$ ($1 \le j \le n$) denote the EET vector of a heterogeneous computing system consisting of a set of machine types, $M = \{M_1, M_2, ..., M_n\}$, and a workload with $c>n$ tasks of type $T_i$ that are all available for execution (such as the bag of tasks \cite{oxley2014makespan}). Tasks are queued upon arrival in a single unbounded FCFS queue. Whenever a machine becomes free, it takes a task from the queue and executes it. Then, the true speed-up is calculated as follows:

   \begin{equation} \label{eq:lem1_sm_avg1}
       \Gamma_M = \widebar{\alpha}^{(i)(A)}                
   \end{equation}

   where 
   
   \begin{equation} \label{eq:alpha_A}         
      \widebar{\alpha}^{(i)(A)} = \frac{1}{n}\cdotp\sum\limits_{j=1}^{n} \alpha^{(i)}_j
   \end{equation}

\end{lemma}

\begin{proof}
We assume that machine type $k$ is the slowest for task type $T_i$, that is, we have $\max\limits_{j=1}^n e_{ij} = e_{ik}$. Then, the baseline system consists of $n$ machines with the expected execution time of $e_{ik}$. For a single FCFS queue, $c$ tasks are equally distributed between homogeneous machines $n$. Hence, each machine has to handle $\frac{c}{n}$ tasks, where the expected execution time of each task is $e_{ik}$. This means that the total time to complete those $c$ tasks in the homogeneous system is $\lceil \frac{c}{n} \rceil \times e_{ik}$. We know $0 \le \lceil \frac{c}{n} \rceil - \frac{c}{n} < 1 $. If we replace $\lceil \frac{c}{n} \rceil$ with $\frac{c}{n}$, the error in calculating the total time is at most $({\lceil \frac{c}{n} \rceil - \frac{c}{n}})/{\lceil \frac{c}{n} \rceil}$, which is negligible for a large number of tasks ($c \gg n$). Thus, for simplicity, we assume that the makepan to complete $c$ tasks in the homogeneous counterpart system is $\frac{c}{n} \times e_{ik}$.

However, in the heterogeneous system, the proportions of the tasks handled by each machine type are not equal, because faster machines can execute more tasks. Specifically, machine type $j$ that has $e_{ij} \le e_{ik}$ executes ${e_{ik}}/{e_{ij}}$ tasks, while machine type $k$ executes only one task. As a result, the proportion of the total number of tasks executed by the slowest machine type to the total number of tasks, denoted by $p_k$, is calculated as follows:

\begin{equation}\label{eq:pk}
    p_k = \Biggl \lceil \frac{c}{\sum\limits_{j=1}^n \frac{e_{ik}}{e_{ij}}} \Biggr \rceil  ,\quad c \ge n
\end{equation}

Then, the makespan of executing $c$ tasks on the heterogeneous system is $p_k \times e_{ik}$. In Equation~\ref{eq:pk}, we know that ${e_{ik}}/{e_{ij}}=\alpha^{(i)}_j$, therefore, the speed-up of the heterogeneous system is calculated as follows:

\begin{equation} \label{eq:speed-up}
    \Gamma_M = \frac{1}{n}\cdotp \frac{c}{\Bigl \lceil \frac{c}{\sum\limits_{j=1}^n \alpha^{(i)}_j} \Bigr \rceil}   
\end{equation}

Assuming that $\frac{c}{\sum\limits_{j=1}^n \alpha^{(i)}_j }\in \mathbb{Z}$, we have $\Biggl \lceil \frac{c}{\sum\limits_{j=1}^n \frac{e_{ik}}{e_{ij}}} \Biggr \rceil = \frac{c}{\sum\limits_{j=1}^n \frac{e_{ik}}{e_{ij}}}$ and Equation~\ref{eq:speed-up} can be represented as follows:

\begin{equation} \label{eq:interval_speed-up}
   \Gamma_M = \frac{1}{n}\cdotp \frac{c}{\frac{c}{\sum\limits_{j=1}^{n} \alpha^{(i)}_j} } = \frac{1}{n}\cdotp \sum\limits_{j=1}^{n} {\alpha^{(i)}_j} =  \widebar{\alpha}^{(i)(A)}
\end{equation}

Note that $\widebar{\alpha}^{(i)(A)}$ is the arithmetic mean of $\Vec{\alpha}^{(i)}$ components. However, if  $\frac{c}{\sum\limits_{j=1}^n \alpha^{(i)}_j }\notin \mathbb{Z}$, the true speed-up is not exactly the arithmetic mean of the row speed-up vector components ($\widebar{\alpha}^{(i)(A)} \neq \Gamma_M$), and the corresponding error, denoted $\epsilon^{(i)}$, is determined as follows:

\begin{equation}
    \epsilon^{(i)} = \frac{\widebar{\alpha}^{(i)(A)}-\Gamma_M}{\Gamma_M} < \frac{\sum\limits_{j=1}^n \alpha^{(i)}_j}{c}
\end{equation}

It should also be noted that for a large number of tasks ($c\gg \sum\limits_{j=1}^n \alpha^{(i)}_j$), we have $\epsilon^{(i)} \approx 0$
\vspace{-5mm}

\end{proof}

In Lemma~\ref{lemma_sm_avg}, we made the assumption that tasks are queued into a single unbounded FCFS queue. Whenever a machine becomes available, it selects a task from the queue and processes it. In support of this scheduling approach, we introduce Lemma~\ref{lemma:FCFS_NQ}, demonstrating that it yields the minimum makespan for bag-of-tasks on heterogeneous computing systems.

\begin{lemma} \label{lemma:FCFS_NQ}
Let $EET=[e_{ij}]$ ($1 \le j \le n$) denote the EET vector of a heterogeneous computing system consisting of a set of machine types, $M = \{M_1, M_2, ..., M_n\}$, and a workload with $c$ tasks of type $T_i$ that are all available for execution (\ie bag of tasks). Then, the minimum makespan (total time to complete the workload) is obtained by using a Round-Robin load balancer across available machines. That is, whenever a machine becomes free, it takes a task from the queue and executes it.
\end{lemma}

\begin{proof}

Based on the Round-Robin load balancer for available machines, whenever a machine becomes free, it takes a task from the queue and executes it. Let $e_k = \max\limits_{j=1}^n e_{ij}$, representing the slowest machine type for the task $T_i$. Based on the proof in Lemma~\ref{lemma_sm_avg}, the number of tasks completed on each machine is as follows:

\begin{equation}
    n_j = \frac{e_k}{e_{ij}} \cdot \frac{c}{\sum\limits_{j=1}^n \frac{e_k}{e_{ij}}} =  \frac{c}{e_{ij} \sum\limits_{j=1}^n \frac{1}{e_{ij}}}
\end{equation}

These fractions have the following characteristics:

\begin{equation} \label{eq:total_c}
    \sum\limits_{j=1}^n n_j = \sum\limits_{j=1}^n \frac{c}{e_{ij} \sum\limits_{j=1}^n \frac{1}{e_{ij}}} = c
\end{equation}

Also, based on Lemma~\ref{lemma_sm_avg}, the makespan, denoted $\tau^*$, is determined as follows:
\begin{equation}
   \forall j  \ \ \ n_j \cdot e_{ij} = \frac{c}{\sum\limits_{j=1}^n \frac{e_k}{e_{ij}}} \cdot e_k = \tau^* 
\end{equation}

We assume that there are other fractions of tasks completed on machines, denoted $n_j^{\prime}$ ($1\le j \le n$), such that the resultant makespan, denoted $\tau^{\prime}$, is less than $\tau^*$. Thus, for each machine, $n_j^{\prime}$ must be less than $n_j$. If there exists a machine with $n_j^{\prime} > n_j$, then the corresponding makespan of the tasks completed on that machine will be $\tau^{\prime} = n_j^{\prime} \cdot e_{ij} > n_j \cdot  e_{ij} = \tau^*$, which is in contradiction with the primary assumption $\tau^{\prime} < \tau^*$. Thus, we have the following:

\begin{equation} \label{eq:njprime}
      n_j^{\prime} < n_j  \ \ \ 1\le j \le n
\end{equation}

Based on Equation~\ref{eq:njprime}, for the total number of tasks completed on the machines, we have:

\begin{equation}
    \sum\limits_{j=1}^n n_j^\prime < \sum\limits_{j=1}^n n_j
\end{equation}

Based on Equation~\ref{eq:total_c}, $\sum\limits_{j=1}^n n_j = c$. As a result,

\begin{equation}
    \sum\limits_{j=1}^n n_j^\prime < c
\end{equation}

Thus, we cannot obtain a makespan less than $\tau^*$ with the same number of tasks. This proves that assigning tasks that are all available for execution on available machines in a Round-Robin manner results in a minimum makespan.
    
\end{proof}

In Lemma~\ref{lemma_sm_avg}, we studied the case where tasks are available for execution a priori so that the impact of task arrival is abstracted. On the other side of the spectrum, we can consider a system with sparse task arrival (\ie large inter-arrival times between tasks) where some machines may become idle while others are busy. In an extreme case with a very low arrival rate, the entire system can potentially become idle. Under certain (threshold) arrival rates, in between the two extremes, only one machine in the system is busy and the rest are idle during the workload processing. Any arrival rate lower than the threshold leads to system idling, whereas any higher arrival rates lead to more than one busy machine at a time. 

\begin{lemma}\label{lemma_sm_h}
    The assumptions are the same as those in Lemma~\ref{lemma_sm_avg}, except that the arrival rate is reduced so that the number of machines busy is exactly one throughout the processing of the workload. In this case, the mean speed-up of using a heterogeneous system for task type $T_i$ is the ``harmonic mean" of the $\Vec{\alpha}$ components, denoted by $\widebar{\alpha}^{(i)(H)}$, and it is calculated as follows:
   
   \begin{equation} 
       \Gamma_M = \widebar{\alpha}^{(i)(H)} = \frac{n}{\sum\limits_{j=1}^m \frac{1}{\alpha^{(i)}_j}}    
   \end{equation}    
\end{lemma}

\begin{proof}
Let $k$ be the slowest machine type in the heterogeneous system that executes task type $T_i$ with the expected execution time of $e_{ik}$. We consider a baseline homogeneous system counterpart with $n$ machines of type $k$ and with the FCFS scheduler. Recall that we assume the task arrival is such that only one machine is busy at any given time. Therefore, the makespan is the sum of execution times for all tasks, and for $c$ tasks the makespan is $c\times e_{ik}$ and each machine executes $\frac{c}{n}$ tasks. Also, the makespan for the heterogeneous system is $\frac{c}{n}\cdotp\sum\limits_{j=1}^m e_{ij}$. Taking these into account, the speed-up of using a heterogeneous computing system is calculated as follows:

\begin{equation} 
       \Gamma_M = \frac{c\cdotp e_{i,k}}{\frac{c}{n}\cdotp\sum\limits_{j=1}^m e_{ij}} = \frac{n}{\sum\limits_{j=1}^m \frac{1}{\alpha^{(i)}_j}}= \widebar{\alpha}^{(i)(H)}       
   \end{equation} 
As we can see, the speed-up of using a machine heterogeneous system that has only one machine busy at any given time aligns with the harmonic mean formula. 
    
\end{proof}

\subsection{Characterizing Task Heterogeneity}
In Section~\ref{sec:charac_machine}, we consider the row speed-up vector to characterize machine heterogeneity for a given task type. Likewise, to characterize the heterogeneity of tasks, for machine type $j$, we define \emph{column speed-up vector}, denoted $\Vec{\beta}_{(j)}$. Then, we summarize the column speed-up vector into a representative mean value, denoted $\widebar{\beta}^H_{(j)}$. 

In Lemmas~\ref{lemma_sm_avg} and~\ref{lemma_sm_h}, based on the true speed-up, we proved that the arithmetic mean (for high task arrival rate) and the harmonic mean (for low task arrival rate) are representative measures for the row speed-up vector due to machine heterogeneity. Next, in Lemma~\ref{lemma_st_avg}, we shift our attention to the task heterogeneity dimension and prove that the harmonic mean is an accurate measure of the central tendency of $\Vec{\beta}_{(j)}$.

\begin{lemma}\label{lemma_st_avg}
    Let $EET=[e_{ij}]$ ($1 \le i \le m$) denote the expected execution time (EET) vector of a set of task types, $T = \{T_1, T_2, ..., T_m\}$ in the machine type $M_j$. A workload trace of $c$ tasks ($c>m$) of type $T_i \in T$ arrive at the system. Tasks are queued upon arrival into a single unbounded FCFS queue and executed by machine $M_j$. Then, the true speed-up is determined as follows:
    

    \begin{equation} \label{eq:lem2_st_avg}
      \Gamma_T = \frac{1}{\sum\limits_{i=1}^{m} \frac{\omega_i}{\beta^i_{(j)}}} = \widebar{\beta}^H_{(j)}       
   \end{equation}
\end{lemma}

\begin{proof}
Assume that the $k$\textsuperscript{th} task type is the largest task type. That is, $e_{kj}$ is the maximum value in the $j^{th}$ column of $EET$. Then, for a homogeneous workload that contains only the type of task $T_k$, the total time consumed by the machine $M_j$ to perform those tasks is $c\times e_{kj}$. However, for a heterogeneous workload with $\omega_i$ as the proportion of each task type to the total number of tasks, the total time required to execute each task type by machine $M_j$ is $c \times \omega_i \times e_{ij}$. Thus, the total time consumed to process all tasks is $\sum_{i=1}^{m}c\cdotp \omega_i\cdotp e_{ij}$. Then, the speed-up of executing the heterogeneous workload, as opposed to the homogeneous workload of type $T_k$, is determined based on Equation~\ref{eq:lem2_st1}.

\begin{equation}\label{eq:lem2_st1}
   \Gamma_T = \frac{c \cdotp e_{kj}}{\sum\limits_{i=1}^{m}c\cdotp \omega_i\cdotp e_{ij}} = \frac{1}{\sum\limits_{i=1}^{m} \omega_i\cdotp \frac{e_{ij}}{e_{kj}}} 
\end{equation}

Based on Equation~\ref{eq:st}, we know that $\beta^i_{(j)} = \frac{e_{kj}}{e_{ij}}$. Additionally, the weighted harmonic mean of the column speed-up vector due to task heterogeneity, denoted by $\widebar{\beta}_{(j)}^H$, is calculated based on Equation~\ref{eq:lem2_st2}.
\begin{equation}\label{eq:lem2_st2}
   \widebar{\beta}_{(j)}^H = \frac{1}{\sum\limits_{i=1}^{m}w_i \frac{1} {\beta^i_{(j)}}} 
\end{equation}

Finally, based on Equations~\ref{eq:lem2_st1} and \ref{eq:lem2_st2}, we prove that the true speed-up due to task heterogeneity is the weighted harmonic mean of $\Vec{\beta_j}$ as noted in Equation~\ref{eq:lem2_conclusion}.
\begin{equation}\label{eq:lem2_conclusion}
    \Gamma_T = \frac{1}{\sum\limits_{i=1}^{m} \omega_i\cdotp \frac{1} {\beta^i_{(j)}}} = \widebar{\beta}_{(j)}^H
\end{equation}    
\end{proof}

\subsection{Homogeneous Equivalent Execution Time (HEET) Measure}
From Lemmas~\ref{lemma_sm_avg} and~\ref{lemma_sm_h}, we learn that, for a high arrival rate, the arithmetic mean represents the mean speed-up ($\widebar{\alpha}^{(i)(A)}$), whereas, for a low arrival rate, harmonic mean ($\widebar{\alpha}^{(i)(H)}$) is representative. However, we are typically interested in studying system efficiency under high arrival rates. That is why, in the rest of this paper, we use the arithmetic mean to represent the speed-up due to machine heterogeneity. Moreover, Lemma~\ref{lemma_st_avg}, proves that the weighted harmonic mean represents the mean speed-up due to task heterogeneity. Taking these into account, as shown in the example of Figure~\ref{fig:summary_method}, we can reduce each column vector $\Vec{\beta}_{(j)}$ to its mean value $\widebar{\beta}_{(j)}^H$ to construct a row vector, denoted $\widebar{\beta}^H = [\widebar{\beta}_{(1)}^H, \widebar{\beta}_{(2)}^H, ..., \widebar{\beta}_{(n)}^H]$, whose contents summarize the execution time behavior of all types of tasks into an equivalent hypothetical task type (denoted $T^*$). We use Equation~\ref{eq:lem2_conclusion} to determine $\widebar{\beta}_{(j)}^H$ for machine type $j$. Note that the mean speed-up due to task heterogeneity for machine type $j$ is calculated with respect to the slowest task type for machine type $j$ (homogeneous counterpart). Therefore, to determine the expected execution time of $T^*$ in the machine type $j$, we consider that $T^*$ is $ \widebar{\beta}_{(j)}^H$ times faster than the slowest task type on machine type $j$. The resultant row vector of the expected execution time of $T^*$ in machine types describes the machine heterogeneity. In a similar approach, we can aggregate machine types into a single hypothetical machine type (denoted $M^*$) whose execution time behavior is representative of the entire set of machine types in the heterogeneous system. To this end, we construct the speed-up vector due to machine heterogeneity for $T^*$, denoted by $\Vec{\alpha}^{(*)}$, based on Equation~\ref{eq:alpha_star}.

\begin{equation} \label{eq:alpha_star}
    \alpha^{(*)}_j = \frac{\max\limits_{i=1}^m e_{ij}}{\widebar{\beta}^H_{(j)}}
\end{equation}

Then, we use the arithmetic mean of $\Vec{\alpha}^{(*)}$ to determine the mean speed-up due to machine heterogeneity, denoted by $\widebar{\alpha}^{(*)(A)}$ for $T^*$. In fact, $\widebar{\alpha}^{(*)(A)}$ represents how much $M^*$ is faster than the slowest machine type for the hypothetical equivalent task type ($T^*$). As a result, we can represent the execution time behavior of the heterogeneous computing system with Homogeneous Equivalent Execution Time, \SEE, using the expected execution time of $T^*$ on $M^*$ as follows:

\begin{equation}\label{eq:measure}    
    \textit{\SEE} = \frac{\max\limits_{j=1}^{n} \alpha^{(*)}_j}{\widebar{\alpha}^{(*)(A)}}
\end{equation}

For the sake of clarity, we use Figure~\ref{fig:summary_method} to illustrate the derivation of \SEE using an example. In Stage (a) of the figure, an EET matrix is considered. Then, in Stage (b), the EET matrix is used to derive the speed-up matrix due to task heterogeneity based on Equation~\ref{eq:st}. Column $j$ of this matrix demonstrates the speed-up due to task heterogeneity for machine type $j$, denoted by $\Vec{\beta}_{(j)}$. Next, in Stage (c), we employ harmonic mean (Equation~\ref{eq:lem2_st2}) to represent each column vector $\Vec{\beta}_{(j)}$ in the form of a scalar value. In this way, the set of speed-up vectors due to task heterogeneity (that constructs a matrix) is reduced to a row speed-up vector. In Stage (d), we calculate the expected execution time of $T^*$ on machines by using the execution time of the slowest task type and $\widebar{\beta}^H_{(j)}$. In Stage (e) we determine the speed-up vector due to machine heterogeneity for $T^*$, based on Equations~\ref{eq:sm}. In stage (f), we use the arithmetic mean to determine the mean speed-up due to machine heterogeneity for $T^*$. Lastly, in stage (g), we calculate the \SEE measure by considering the speed-up value obtained in stage (f) considering the execution time of the slowest machine type for $T^*$ ($M_1$ in stage (d)) as the baseline.

\begin{figure}[th]
    \centering
    \includegraphics[width=0.5\textwidth]{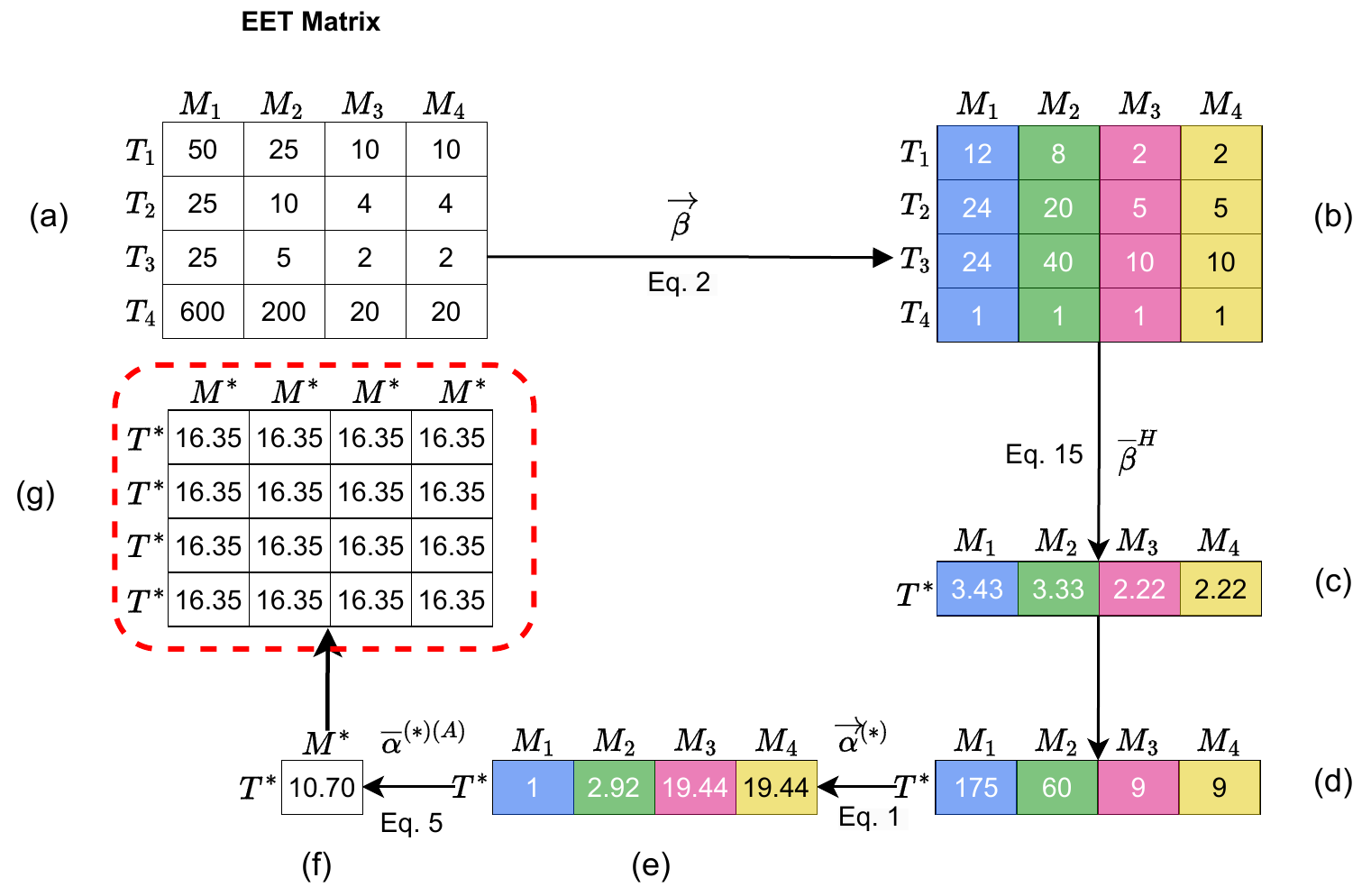}
    \vspace{-8mm}
     
    \caption{An example illustrating the stages to calculate the heterogeneity measure (\SEE). (a) The EET matrix representing a heterogeneous system. (b) The speed-up matrix due to task heterogeneity is derived from EET matrix based on Equations~\ref{eq:st}. (c) Based on Equation~\ref{eq:lem2_st2}, we consolidate each column into a mean value using the harmonic mean. (d) Calculate the expected execution time of $T^*$ on machines using the execution time of the slowest task type and $\overline{\beta}^H_j$. (e) The speed-up vector due to machine heterogeneity for $T^*$, based on Equations~\ref{eq:sm}. (f) we employ Equation \ref{eq:alpha_A} on the resultant speed-up vectors to determine the mean speed-up value due to machine heterogeneity. (g) \SEE score represents the execution behavior of the heterogeneous system.}    
    \label{fig:summary_method}
    \vspace{-2mm}
\end{figure}
While \SEE is able to accurately characterize the system heterogeneity, our hypothesis is that, for a given workload, systems with similar \SEE scores exhibit similar makespan too. 
For a given workload trace with $T$ task types, system $A$ with the set of heterogeneous machine types $M_A$ offers a bigger makespan than heterogeneous system $B$ with machine types $M_B$ ($|M_A|=|M_B|$), if and only if $\SEE_A > \SEE_B$.

\subsection{\SEE: Space Mapping From Heterogeneity to Homogeneity}
As shown in Figure~\ref{fig:summary_method}, the mathematical approach employed to obtain \SEE is actually to transform the EET matrix that represents the heterogeneous system (stage (a)) into a homogeneous EET matrix whose elements are values \SEE (stage (g)). In other words, we proposed a mathematical formulation that transforms a heterogeneous computing system into a hypothetical homogeneous system such that both perform similarly in terms of performance metrics such as makespan. 

Recall that the \SEE metric is the expected execution time of the hypothetical equivalent task type ($T^*$) on the hypothetical equivalent machine type ($M^*$). Replacing the machine types with $M^*$ and task types with $T^*$, results in a homogeneous EET matrix, denoted by $EET^*$, whose elements are \SEE values. The lemmas~\ref{lemma_sm_avg} and \ref{lemma_st_avg} prove that the homogeneous system represented by the matrix $EET^*$ exhibits a similar makespan as the heterogeneous system represented by the EET matrix. Given a workload of $c$ of tasks, its makespan, denoted by $\tau$, on a heterogeneous computing system can be estimated by the makespan of the same workload on the homogeneous equivalent system represented by $EET^*$. For a homogeneous system of $n$ machines ($M^*$), the expected makespan of executing $c$ tasks of the same type ($T^*$) is the number of tasks distributed on each machine ($\frac{c}{n}$) multiplied by the expected execution time of that task type on the machine type (\ie \SEE value).  As a result, the makespan is determined as follows:
\begin{equation} \label{eq:tau}
    \tau = \frac{c}{n} \cdot \textit{\SEE}
\end{equation}

Similarly, we can use the \SEE score to estimate the system throughput. To do this, we estimate the throughput metric, denoted by $\theta$,  as the ratio of the number of tasks to the time it takes to complete those tasks on the machines. Thus, the throughput is determined as follows:

\begin{equation} \label{eq:throughput}
    \theta = \frac{n}{\textit{\SEE}}   
\end{equation}

\section{System Design}
\label{system-design}

In this section, we present an overview of the components that we have designed to evaluate the \SEE metric. Our implementation includes a real-world end-to-end ``inference system'', customized to match real-world production scenarios \cite{salesforce-blog-2020}. 
Throughout our experimentation, we used AWS EC2 instances as machines. Note that, with slight modifications, the system can also be readily deployed on alternative cloud platforms. 
The primary objective of our system is to validate the precision of the estimated makespan based on Equation~\ref{eq:tau} by comparing them against the actual makespan in various configurations of the system. Figure~\ref{fig:system_design} illustrates the overall architecture of the system. The components of the system are explained in the following paragraphs.

\noindent \textbf{Model Loader} To encompass a diverse range of model types, we used the extensive model repository offered by HuggingFace \cite{wolf2019huggingface}. We used four different models in our experiment: (1) Image classification: Resnet50 \cite{he2016deep}, (2) Object detection: Yolov5 \cite{Jocher_YOLOv5_by_Ultralytics_2020}, (3) Question answering: DistilBERT \cite{sanh2019distilbert}, and (4) Speech recognition: Wav2vec2 \cite{baevski2020wav2vec}. Given the variety of deep learning frameworks from which these models were sourced, we sought consistency in our experiments. To achieve this, we converted all models to the ONNX (Open Neural Network Exchange) format \cite{onnx-github} using the PyTorch ONNX converter \cite{onnx-conversion}. Utilizing ONNX grants us the advantage of a unified model server setup, applicable across all model types.

\noindent \textbf{Model Server} Each of the models used during the experiments should be deployed as a machine learning service. We have used the multi-model serving capability of modern inference systems \cite{multi-model-serving} to encapsulate multiple models into a single inference service. Each of the services is supported by an AWS EC2 instance. On each machine, we spin up a containerized version of NVIDIA Triton Inference Server \cite{triton-inference-server} and load all model variants onto it. Communications to model servers are implemented using gRPC \cite{grpc} due to its superior performance compared to other transfer protocols \cite{grpc-vs-rest}.

\noindent \textbf{Workload}  We synthesized the workload traces assuming that the interarrival time between tasks in the workload traces follows an exponential distribution with the mean arrival rate as its parameter~\cite{harchol2013performance}. According to the bag-of-task assumption \cite{mokhtari2022felare}, we have also designed a workload of tasks that are all available for execution from the beginning (\ie arrival time is zero). Tasks are sent asynchronously to the machines that host the ML services based on their arrival times.

\begin{figure}
    \centering
    \includegraphics[width=0.49\textwidth]{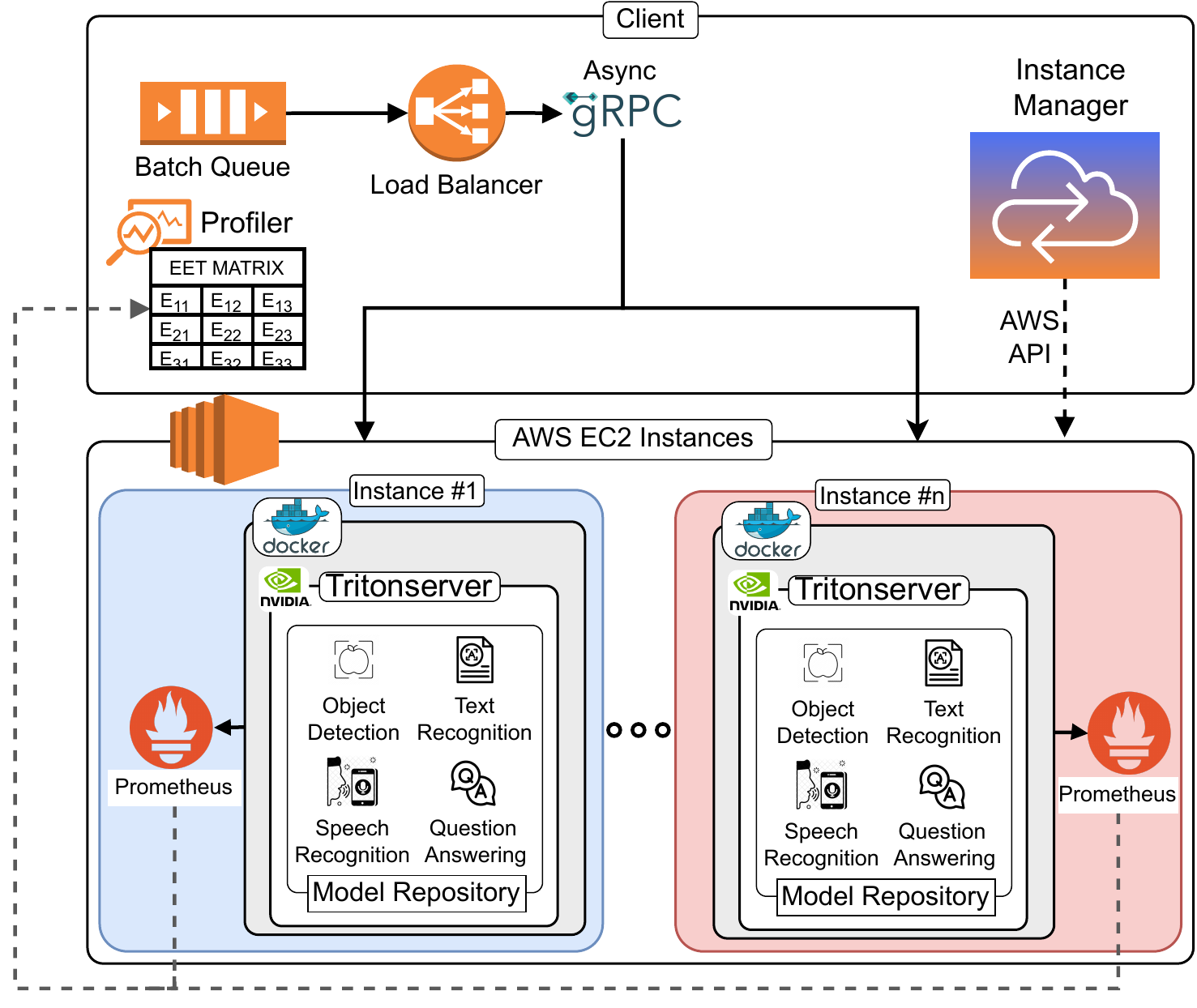}
    \vspace{-5mm}
    \caption{Performance comparison of the same set of scheduling methods with the same workload across two computing systems, A and B, with different levels of heterogeneity (horizontal axis). The vertical axis shows the percentage of tasks completed on time.}
    \label{fig:system_design}
    \vspace{-6mm}
\end{figure}

\noindent \textbf{Monitoring} The monitoring component in each of the EC2 instances is equipped with Prometheus \cite{prometheus}, a monitoring system and a time series database. It records the inference time of all model servers during experiments and is stored in the database for later analysis. To construct the EET matrix, Profiler benchmarks each task type with the machine types and employs Prometheus to obtain the expected inference times.

\noindent \textbf{Instance Manager} During experiments, it is necessary to reconfigure heterogeneous computing systems with a different number of instances of each type. To support this, we have automated the process of reconfiguring the system in a central instance manager. The process of transitioning between two instance configurations includes (1) removing the current instances and cleaning the cluster, (2) setting up the new sets of machines with required dependencies, (3) bringing up the Triton container on top of the EC2 instances, and finally loading the models to the Triton inference server. We have automated all steps 1-4 using the AWS Python SDK \cite{aws-sdk-python}. Furthermore, the types of EC2 instance that we used during our experiments are (1) t2.large, (2) c5.2xlarge and (3) g4dn.xlarge. These instance types are similar to slow, medium, and fast machines in inferring the selected machine learning tasks. 

\noindent \textbf{Load Balancer}
Tasks are queued in a single unbounded FCFS queue upon arrival. Then, the load balancer assigns tasks to available machines in a round-robin manner. Mapping events are triggered by the completion or arrival of a task. In case a task arrives at the system while there is no available machine, the load balancer defers the mapping event until a machine becomes free.

\section{Experimental Validation} \label{experiments}

\subsection{Experimental Setup} \label{exp_setup}

To validate the developed heterogeneity measure, in this section, we execute a workload of four deep learning applications on various combinations of three types of Amazon EC2 instances (machine) to verify the \SEE score in real-world settings. Specifically, we used four different applications/models in our experiment: (1) image classification implemented using the Resnet50 model \cite{he2016deep}, (2) object detection implemented according to the Yolov5 model \cite{Jocher_YOLOv5_by_Ultralytics_2020}, (3) question answering based on the DistilBERT model \cite{sanh2019distilbert}, and (4) speech recognition using the Wav2vec2 model \cite{baevski2020wav2vec}. For machine types, we utilize GPU-based (\texttt{g4dn.xlarge}), compute-optimized (\texttt{c5.2xlarge}), and general-purpose CPU-based (\texttt{t2.large}) Virtual Machines offered by AWS EC2 services. To obtain the expected execution time of the image classification task on these machines, we processed 1000 sample images on each instance type. We repeat the experiment 10 times, and finally, we use the expected value of these 10,000 inference operations in the EET matrix. Similarly, for the object recognition task, we ran the object recognition task for 1000 sample images 10 times to determine the expected execution time of the object recognition task. For the speech recognition task type, we execute a recorded audio of length 4 seconds on the machine types. Then, the average value of the inference times is used to fill the EET matrix. For question answering, we provide a sample context and question as input of the inference task and run the inference task 1000 times.  We aggregated the inference times and determined the expected execution time for use in the EET matrix. 

To synthesize the workload trace, we assume that the inter-arrival time between tasks in the workload traces follows an exponential distribution with the mean arrival rate as its parameter~\cite{harchol2013performance}. For the bag-of-tasks, we assume that all tasks are available for execution from the beginning (arrival time is zero). Tasks are queued in a single unbounded FCFS arrival queue upon arrival, and assigned to the available machines in the round-robin manner. Tasks are considered to be latency-sensitive with hard deadlines. The performance metric (makespan) is defined as the time that the system requires to complete all tasks in the workload trace.

Recall that we employed hybrid central tendency measures (column-wise harmonic mean and row-wise arithmetic mean) to obtain a meaningful representative of the execution time behavior of heterogeneous computing systems. To illustrate the effectiveness of our method, we compare the experimental results with the following baselines as representatives of the execution time behavior of heterogeneous systems: (1) Arithmetic mean, (2) Harmonic mean, and (3) Geometric mean of the elements of the EET matrix. 





\subsection{Validating \SEE Score as a Performance-Driven Heterogeneity Measure} \label{sec:exp1}
In this experiment, our goal is to assess the ability of the \SEE  score to estimate the true throughput and the true makespan of the heterogeneous computing systems using Equations \ref{eq:throughput} and \ref{eq:tau}, respectively. To study the behavior of heterogeneous systems with varying \SEE scores, we simulate a variety of heterogeneous systems with three types of machines and four types of tasks. To this end, we generate 228 different system configurations with a different number of instances of each type (\ie \texttt{t2.large}, \texttt{c5.2xlarge}, and \texttt{g4dn.xlarge}). Then, for each system configuration, we feed it with a bag of 1000 tasks of four types (\ie image classification, object detection, question answering, and speech recognition). Eventually, we average the makespan and throughput on different system configurations with the same \SEE score. 

\begin{figure*}[!t]
\centering
\subfloat[]{\includegraphics[width=0.41\textwidth]{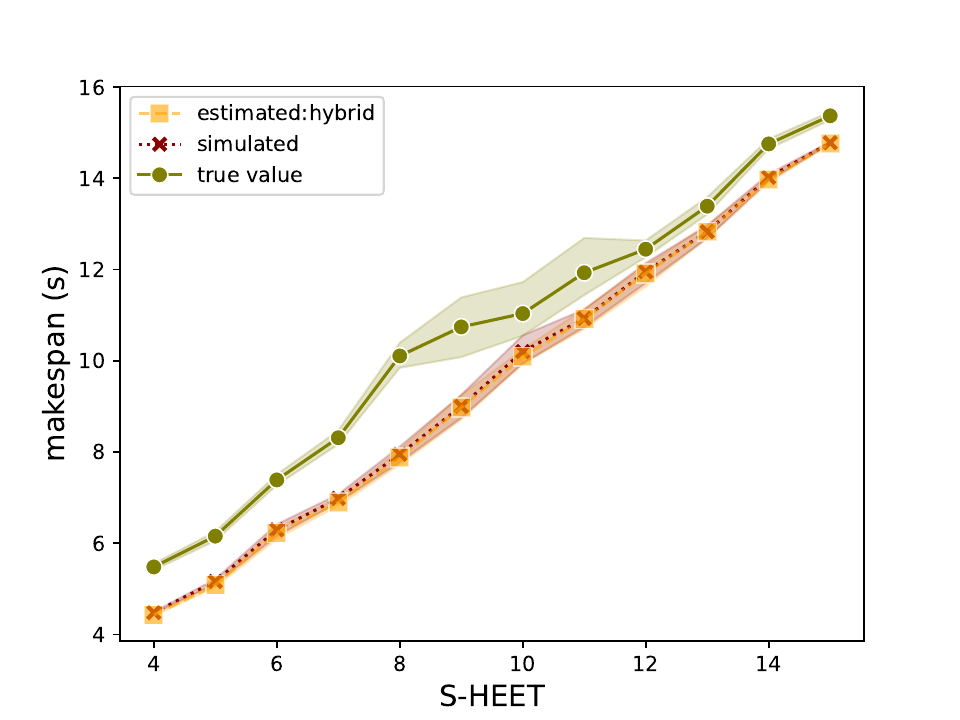}%
\label{fig:see_makespan_mixed}}
\hfil
\subfloat[]{\includegraphics[width=0.41\textwidth]{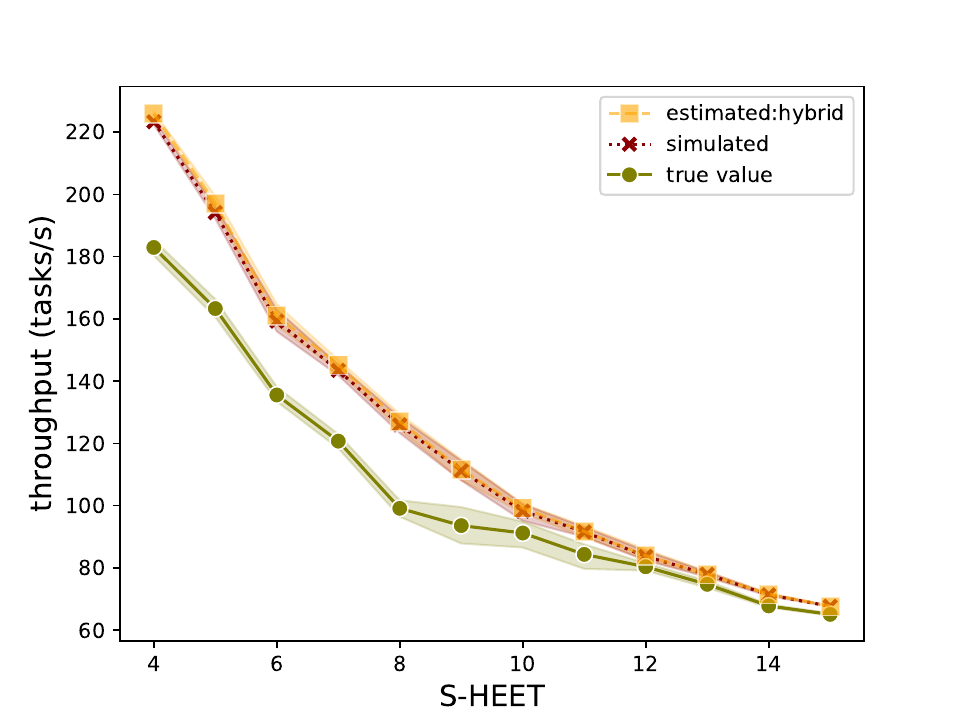}%
\label{fig:see_throughput_mixed}}
\caption{The system performance in terms of the makespan (left) and throughput (right) (vertical axis) upon varying S-\SEE  scores (horizontal axis) for workloads 1000 tasks. Each individual point represents the average result of multiple computing systems with the same S-\SEE  score. Furthermore, the colored area illustrates the 95\% confidence interval of the results.}
\label{fig:see_mixed}
\vspace{-6mm}
\end{figure*}


Figure~\ref{fig:see_mixed} shows the results of the true and estimated makespan and throughput with varying S-\SEE score, which is the \SEE score scaled with the number of machines (\ie $\frac{\SEE}{n}$). Each (makespan, S-\SEE)/(throughput, S-\SEE)  point shows the average makespan/throughput over different system configurations with the same S-\SEE value. The colored area illustrates the 95\% confidence interval of the results. As shown in Figure~\ref{fig:see_mixed}, we can observe that systems with lower S-\SEE scores generally perform better (\ie lead to a smaller makespan or higher throughput) than those with higher \SEE score values. This statement itself means that the S-\SEE score can be effectively used as a measure to ``\emph{compare different heterogeneous computing systems}" in terms of makespan or throughput. Moreover, we can observe that the results exhibit a narrow confidence interval, that is, systems with the same S-\SEE score will perform similarly in terms of makespan and throughput. The number of different configurations with S-\SEE scores equal to 9, 10, or 11 is small (less than 9), therefore we observe a wider confidence interval for these S-\SEE scores in the results. 

In addition, the results show that the estimated values (makespan and throughput) using the \SEE score based on Equations \ref{eq:throughput} and \ref{eq:tau}, respectively, can predict the true values with an average accuracy of 84\%. Note that we used the expected values of the execution times to determine the \SEE score. As a result, the accuracy of the estimation depends on the degrees of uncertainty that exist in the execution times. Accordingly, we ran a simulation with zero uncertainty in expected execution times to demonstrate the root cause of the error in estimating the makespan and throughput using the \SEE score. The results show that the makespan calculated using the \SEE score accurately estimates the makespan of the simulation. Thus, we can say that in heterogeneous systems with low levels of uncertainty in execution times, estimating makespan using \SEE score is an accurate method. 

Given a user-defined throughput threshold, we can determine the corresponding \SEE  score and use it to configure a heterogeneous system with the desired throughput. For a desired throughput, the \SEE  score enables solution architects to proactively configure a heterogeneous system that can meet that objective (instead of try and error) with minimum cost.

In summary, the result of this experiment validates the applicability of the \SEE score for real-world scenarios. In particular, when the \SEE score is applied across systems, it can accurately compare different heterogeneous systems with respect to their performance metrics (makespan and throughput) without examining the workload in these systems. 

\subsection{Correctness of arithmetic, harmonic, and geometric means as heterogeneity measures}
In this experiment, we investigate the effectiveness of the arithmetic, harmonic, and geometric mean of expected execution times of task types in machine types as heterogeneity measures representing the execution behavior of heterogeneous systems. To this end, we conducted a similar experiment as in Section~\ref{sec:exp1} to study the performance (in terms of makespan) of 228 heterogeneous computing systems. An appropriate heterogeneity measure should be able to identify similarity and superiority in terms of performance metrics (\eg makespan or throughput) across heterogeneous systems. 

\begin{figure*}[!t]
\centering
\subfloat[]{\includegraphics[width=0.33\textwidth]{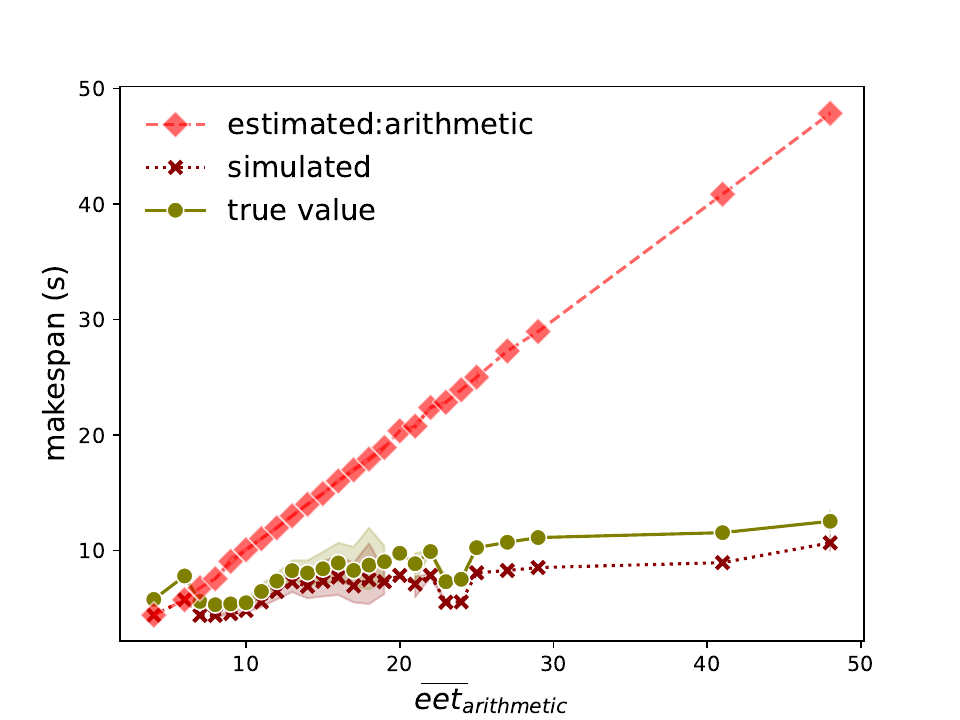}%
\label{fig:arithmetic}}
\hfil
\subfloat[]{\includegraphics[width=0.33\textwidth]{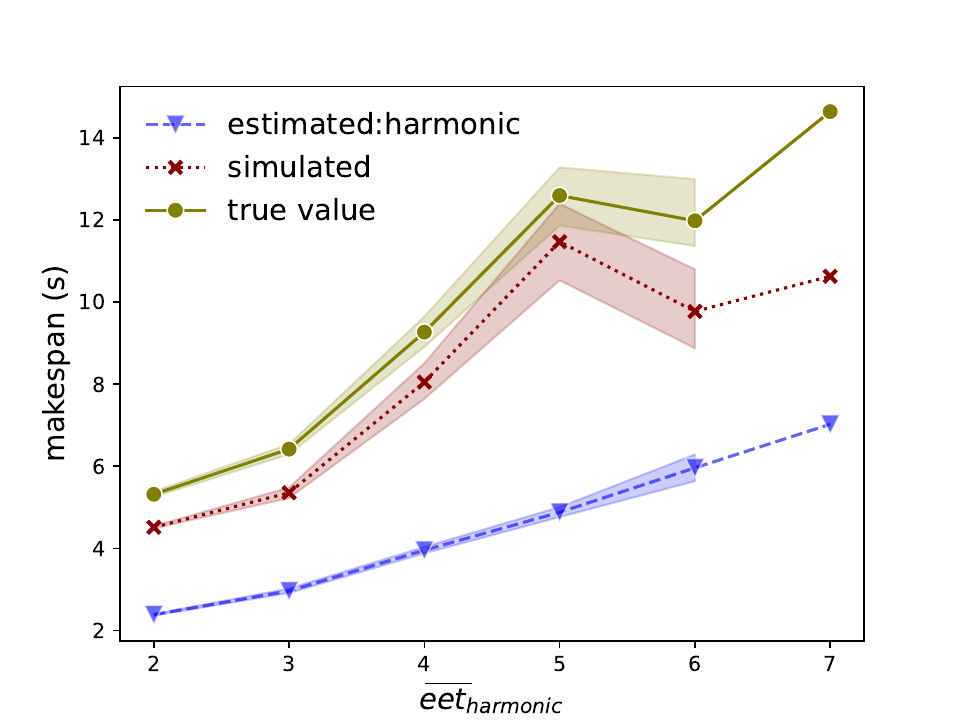}%
\label{fig:harmonic}}
\hfil
\subfloat[]{\includegraphics[width=0.33\textwidth]{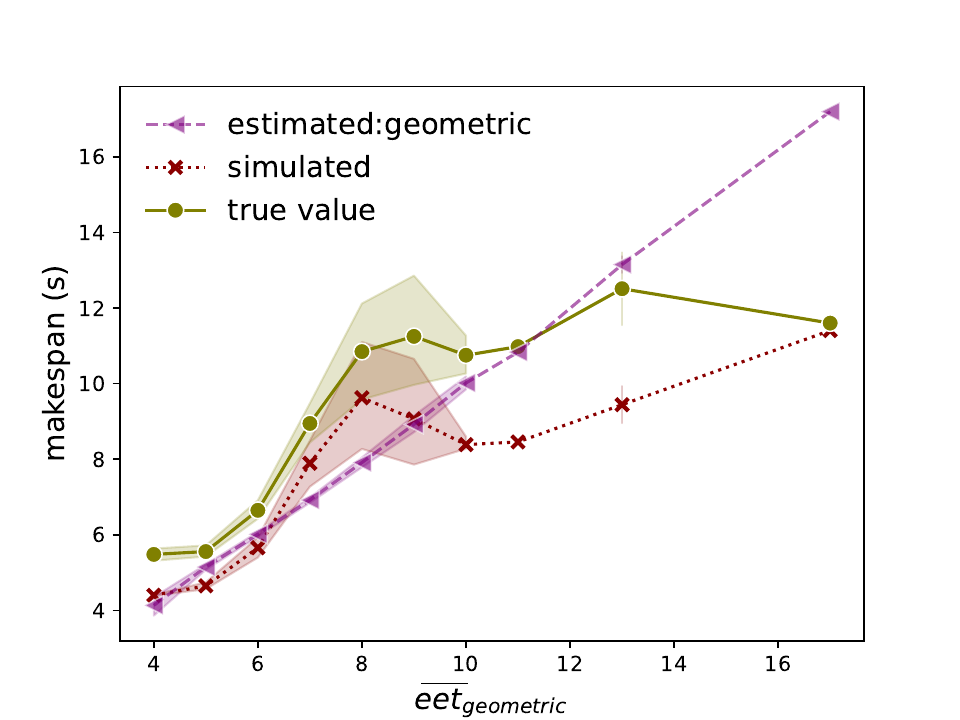}%
\label{fig:geometric}}

\caption{The system performance in terms of the makespan (vertical axis) with respect to the varying mean of the EET matrix (horizontal axis) for a workload of 1000 tasks. In (a)--(c), $\overline{eet}_{\textrm{arithmetic}}$, $\overline{eet}_{\textrm{harmonic}}$, and $\overline{eet}_{\textrm{geometric}}$ are the arithmetic, harmonic, and geometric mean of the expected execution times of task types on machines types as represented in the EET matrix. Each individual point represents the average result of multiple computing systems with the same mean EET value. Additionally, the colored area illustrates the 95\% confidence interval of the results.}
\label{fig:baselines}
\vspace{-6mm}
\end{figure*}

     

Figure~\ref{fig:baselines} shows the results of the makespan for different heterogeneous systems with varying mean expected execution times calculated based on
arithmetic ($\widebar{eet}_{arithmetic}$), harmonic ($\widebar{eet}_{harmonic}$), and geometric ($\widebar{eet}_{geometric}$) mean techniques.
In Figure~\ref{fig:arithmetic}, the results show that the estimated makespan using $\widebar{eet}_{arithmetic}$ cannot follow the true makespan and is considerably inaccurate. Similarly, we observe that the harmonic mean and geometric mean heterogeneity measures are also inaccurate.

In summary, comparing the results for the baseline heterogeneity measures, as shown in Figure~\ref{fig:baselines}, with the results for the \SEE score, as shown in Figure~\ref{fig:see_mixed}, verifies that heterogeneous systems are well characterized with the \SEE measure and can be used to accurately estimate the performance of heterogeneous systems.

\section{Related works}\label{related_works}
Heterogeneous computing systems utilize various computing machines to perform diverse tasks with different computational requirements. The idea of exploiting system heterogeneity to improve system performance considering different objectives such as energy~\cite{hussain2021energy,panda2019energy,chen2020joint, ghafouri2020consolidation, mokhtari2022felare} and QoS~\cite{azizi2022deadline, hussain2022hybrid, mokhtari2020autonomous, denninnart2020efficient} has been extensively explored in the literature. Based on these works, it is proven that heterogeneity can play a crucial role in enhancing different system performance metrics; however, these works fall short of providing a concrete metric that can explain the impact of heterogeneity on system performance. Moreover, these works commonly try to exploit heterogeneity in favor of optimizing a performance metric. In contrast, our work takes a different approach such that for a desired performance metric it tries to configure the heterogeneity of the system and its impact on the performance metric.

\textbf{Expected Time to Compute (ETC) Matrix} The idea of characterizing a heterogeneous computing system using the Expected Time-to-Compute (ETC) matrix was first explored by Ali et al.~\cite{ali2000representing}. They used the coefficient-of-variation of expected execution times as a measure of heterogeneity. Then, they suggested an algorithm that takes the mean and standard deviation of execution times to generate the ETC matrix of the heterogeneous system. However, their method neither characterizes the performance of different heterogeneous systems nor makes them comparable.  In contrast, we present a mathematical model to measure system heterogeneity such that it can characterize the overall system performance behavior and make different systems comparable.

\textbf{Heterogeneity-aware Task Scheduling} Several research works have been carried out on heterogeneity-based scheduling algorithms. Panda et al.~\cite{panda2019energy} introduced an energy efficient task scheduling algorithm, called ETSA, for heterogeneous cloud computing to minimize energy consumption and makespan (\ie the total time to complete a workload). In the proposed algorithm, the trade-off between minimizing energy consumption and maximizing system performance is achieved by minimizing the linear combination of normalized completion time and total utilization. In their work, the heterogeneity is represented by the Expected Time to Compute (ETC) matrix whose entries show the expected execution time of a particular task type on a specific machine type. Although the proposed scheduling solution has been evaluated in different heterogeneous computing systems, the impact of heterogeneity on system performance has not been discussed. 
Moreover, Mokhtari \etal~\cite{mokhtari2022felare}, leveraged the ETC matrix to model the performance behavior of heterogeneous computing systems. They introduced a fair and energy-aware scheduling algorithm, called FELARE, for latency-sensitive tasks in heterogeneous edge systems. To improve fairness across task types, FELARE monitors the performance of the heterogeneous system, and based on the defined fairness metric, it mitigates the suffered tasks by prioritizing them in the next mapping events.

In~\cite{denninnart2020efficient, mokhtari2020autonomous}, the authors followed a probabilistic approach to design a robust resource allocation algorithm for heterogeneous computing systems. 
In these works, a modified version of the ETC matrix was employed that contains the distribution of execution time for each task type on a machine type to model the heterogeneity.
Narayanan \etal~\cite{narayanan2020heterogeneity} proposed a throughput matrix to model the performance behavior of the system. Specifically, each entry $(i,j)$ in the matrix represents the performance of the job $i$ on the machine $j$. The matrix also implies the system heterogeneity; hence, they leveraged it to devise a heterogeneity-aware scheduling method that can be optimized for different performance metrics.

\textbf{Heterogeneity-aware Machine Learning Inference Serving Systems} 
Several research works have been conducted to devise heterogeneity-aware machine learning inference services considering performance objectives such as cost, QoS, or throughput. Ribbon \cite{ribbon} takes advantage of the cost and performance trade-off to serve deep learning inferences in heterogeneous AWS EC2 instances. They considered different system configurations (different numbers of AWS EC2 instances) as decision parameters. In their work, they assume that the performance metric (\ie QoS) of diverse configurations cannot be described mathematically. Thus, they employed Bayesian optimization techniques to find the optimal configuration that minimizes the cost of serving inference queries while meeting the QoS constraint. Cocktail \cite{cocktail} uses the modeling that ensembles heterogeneous model variants along with heterogeneity in hardware to improve accuracy and cost optimization. Kairos \cite{li2022kairos} is a deep learning inference serving framework that maximizes throughput under the cost budget and QoS constraint. For that purpose, they proposed a heterogeneity-aware query distribution mechanism that maximizes throughput. The core idea is to distribute queries across machines so that it maximizes the idle time in the future for all instances. In these works, it is assumed that the performance (\ie throughput) of a heterogeneous system cannot be mathematically described; however, in our work, we propose an analytical approach to accurately estimate the throughput of a heterogeneous system. This would help to find the optimal configuration that minimizes the cost of service while meeting the throughput target.

In sum, all these works testify that heterogeneity can be employed to improve the system's performance. However, there is yet to be a concrete way to characterize the impact of heterogeneity on the performance of the system. For this purpose, we provide a heterogeneity measure that characterizes the impact of heterogeneity on the performance of heterogeneous systems and can be used to estimate the performance (\ie throughput or makespan) of heterogeneous systems without online evaluations.

\section{Conclusion} \label{conclusion}

In this research, we provided a measure to quantify the heterogeneity of the system with respect to the performance metric (makespan or throughput) of the system and for a given set of task types. We characterize system heterogeneity by decoupling the heterogeneity into machine and task dimensions. To quantify the performance impact of each heterogeneity dimension, we devised a speedup vector due to machine heterogeneity and task heterogeneity. We proved that the mean speedup due to machine heterogeneity under high workload arrival is measured by the arithmetic mean. We also proved that harmonic mean can measure mean speedup due to task heterogeneity. 
Then, we leveraged the mean speedup due to task heterogeneity to reduce all task types to a hypothetical equivalent task type that can characterize the execution time behavior of the set of task types. Next, we employ the mean speedup due to machine heterogeneity to characterize the set of machine types into a single hypothetical equivalent machine type. 
Finally, we introduce the \SEEL \xspace (\SEE) score as a heterogeneity measure representing how fast a heterogeneous system is for a given set of task types. In this way, we transform a heterogeneous computing system into a hypothetical equivalent homogeneous system with similar performance metrics (makespan and throughput). \SEE can be used to globally compare the performance of different heterogeneous systems. We observe that the \SEE score is effective and can approximate the makespan and throughput with an average and minimum accuracy of 84\% and 80.0\%, respectively. For a desired throughput objective, the \SEE score enables solution architects to proactively configure a heterogeneous system (instead of try and error) that can fulfill that objective. In the future, we plan to incorporate probabilistic analysis into our analysis to capture uncertainties in execution times. Another avenue for future work is to devise a task scheduling metaheuristic based on the \SEE score.



\bibliographystyle{IEEEtran}
\bibliography{references}


\begin{IEEEbiography}[{\includegraphics[width=1in,height=1.25in,clip,keepaspectratio]{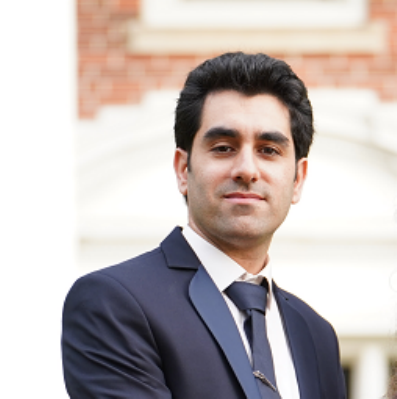}}]{Ali Mokhtari}
received the B.Sc. degree in mechanical engineering from Shiraz University, Iran, in 2010, the M.Sc. degree in aerospace engineering from the Sharif University of Technology, Tehran, Iran, in 2016, and the M.Sc. and Ph.D. degrees in computer science from the University of Louisiana at Lafayette, USA, in 2023. His research focuses on heterogeneous computing systems, including cost- and QoS-aware design of distributed systems and heterogeneity-aware resource allocation in distributed systems. \end{IEEEbiography}

\begin{IEEEbiography}[{\includegraphics[width=1in,height=1.25in,clip,keepaspectratio]{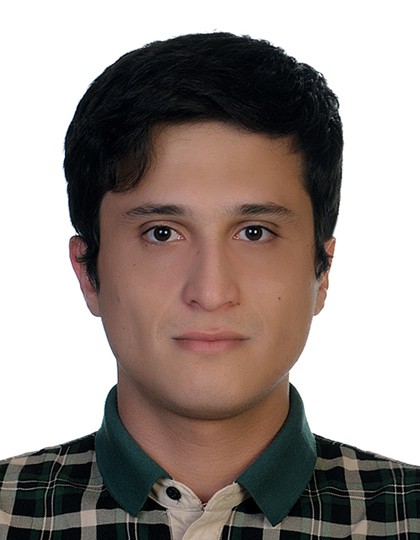}}]{Saeid Ghafouri}
holds a B.Sc. in computer software engineering from Shahid Chamran University of Ahvaz, Iran, and an M.Sc. in Artificial Intelligence from the Computer Engineering Department at K. N. Toosi University of Technology, Tehran, Iran. Currently pursuing a Ph.D. in computer science at Queen Mary University of London, his research focuses on cloud and edge computing, resource allocation in these domains, and reinforcement learning. Saeid recently completed a one-year internship at the AISys lab, University of South Carolina.
\end{IEEEbiography}

\begin{IEEEbiography}[{\includegraphics[width=1in,height=1.25in,clip,keepaspectratio]{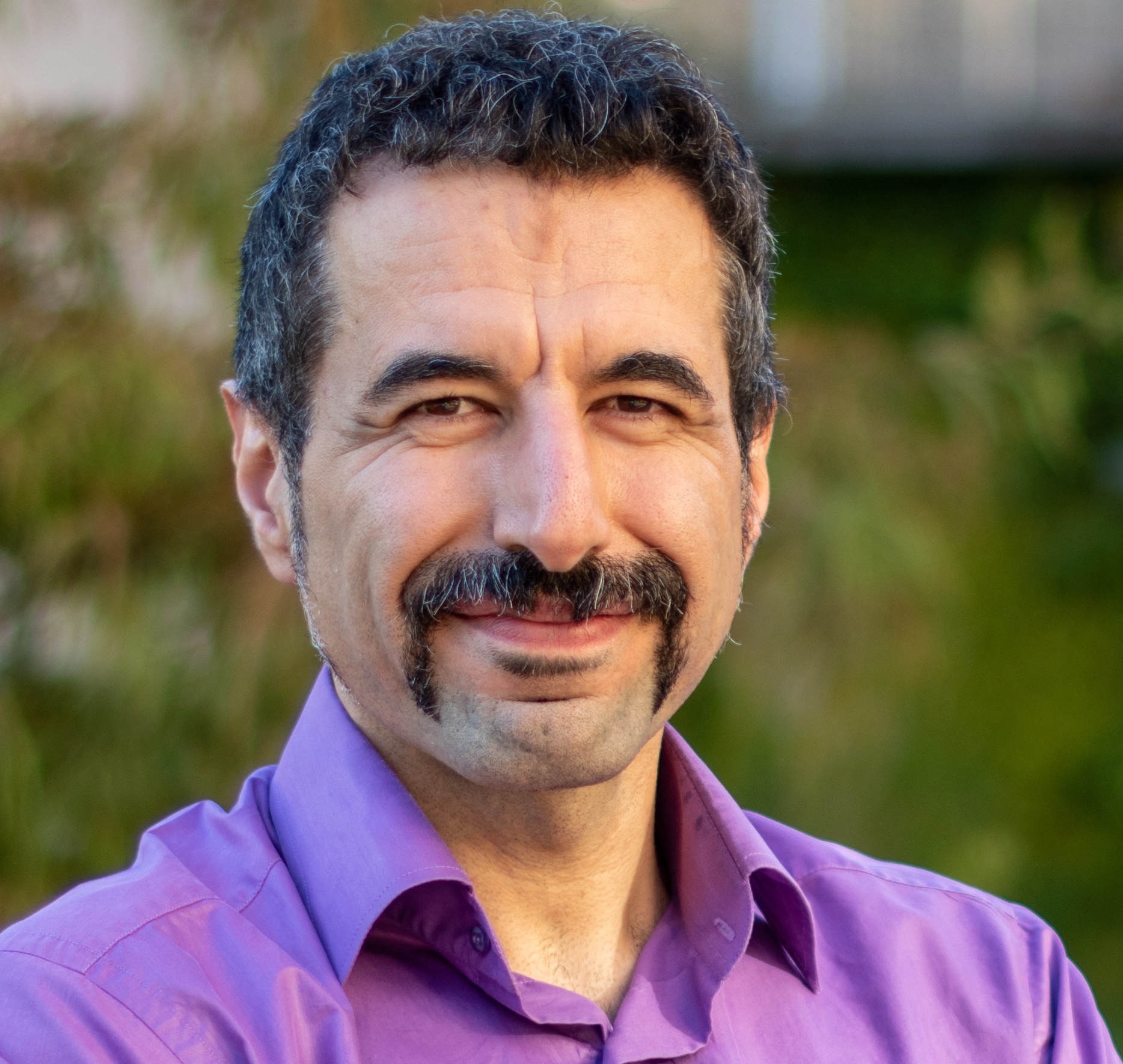}}]{Pooyan Jamshidi}
is an Assistant Professor in Computer Science at UofSC and a visiting researcher at Google. He directs the AISys Lab, where he develops theory in Causal AI and Statistical ML and applies the theories to solve problems in Computer Systems, Robot Learning, and Software Engineering. He is, in particular, interested in Causal Representation Learning, Transfer Learning, Reinforcement Learning, ML Security/Explainability, and ML for Systems. Prior to his current position, he was a research associate at CMU and Imperial. He received a Ph.D. in Computer Science at DCU in 2014 and M.S. and B.S. degrees in Systems Engineering and Computer Science from the AUT in 2003 and 2006.
\end{IEEEbiography}

\begin{IEEEbiography}[{\includegraphics[width=1in,height=1.25in,clip,keepaspectratio]{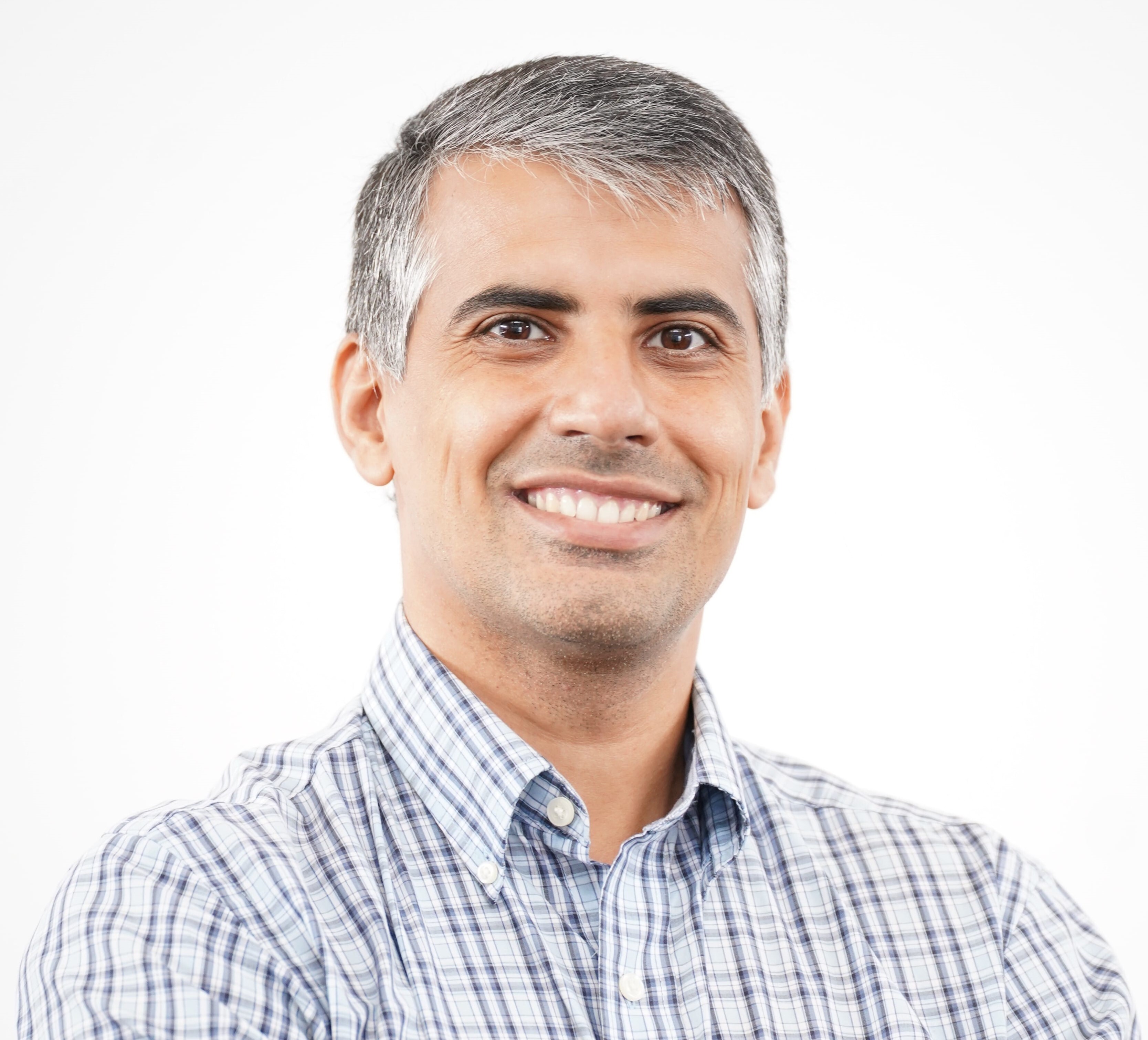}}]{Mohsen Amini Salehi}
received his Ph.D. in Computing and Information Systems from Melbourne University, in 2012. He is an NSF CAREER Awardee Associate Professor and the director of HPCC lab, at the department of Computer Science and Engineering, University of North Texas. His research focus is on the next generation of cloud computing systems, including edge-to-cloud continuum, heterogeneity, virtualization, and cloud AI.
\end{IEEEbiography}
\vfill
\end{document}